\documentclass[sn-mathphys-num]{sn-jnl}

\usepackage{graphicx}%
\usepackage{multirow}%
\usepackage{amsmath,amssymb,amsfonts}%
\usepackage{amsthm}%
\usepackage{mathrsfs}%
\usepackage[title]{appendix}%
\usepackage{xcolor}%
\usepackage{textcomp}%
\usepackage{manyfoot}%
\usepackage{booktabs}%
\usepackage{algorithm}%
\usepackage{algorithmicx}%
\usepackage{algpseudocode}%
\usepackage{listings}%
\usepackage{cleveref}

\theoremstyle{thmstyleone}%
\newtheorem{theorem}{Theorem}
\newtheorem{lemma}{Lemma}
\newtheorem{corollary}{Corollary}
\newtheorem{definition}{Definition}%

\theoremstyle{thmstyletwo}%

\theoremstyle{thmstylethree}%

\begin{document}

\title[A P Algorithm to sovle ARP]{A Polynomial-time Algorithm to Solve the Airplane Refueling Problem: the Sequential Search Algorithm}


\author*[1]{\fnm{Jinchuan} \sur{Cui}}\email{cjc@amss.ac.cn}

\author[2]{\fnm{Xiaoya} \sur{Li}}\email{xyli@amss.ac.cn}


\affil[1,2]{\orgdiv{Academy of Mathematics and Systems Science}, \orgname{Chinese Academy of Sciences}, \orgaddress{\street{No. 55 Zhongguancun East Road}, \city{Beijing}, \postcode{100190}, \country{China}}}



\abstract{Airplane refueling problem is a nonlinear unconstrained optimization problem with $n!$ feasible solutions. Given a fleet of $n$ airplanes with mid-air refueling technique, the question is to find the best refueling policy to make the last remaining airplane travel the farthest. In order to solve airplane refueling problem, we proposed the definition of sequential feasible solution by employing the refueling properties of data structure. We proved that if an airplane refueling instance has feasible solutions, it must have sequential feasible solutions; and the optimal feasible solution must be the optimal sequential feasible solution. So we need to numerate all the sequential feasible solutions to get an exact algorithm. We proposed the sequential search algorithm which consists of two steps, the first step of which aims to seek out all of the sequential feasible solutions, and the second step aims to search for the maximal sequential feasible solution by bubble sorting all of the sequential feasible solutions. We observed that the number of the sequential feasible solutions will change to grow at a polynomial rate when the input size of $n$ is greater than an inflection point $N$. Then we proved that the sequential search algorithm is a polynomial-time algorithm to solve the airplane refueling problem. Moreover, we built an efficient computability scheme, according to which we could forecast within a polynomial time the computational complexity of the sequential search algorithm that runs on any given airplane refueling instance. Thus we could provide a computational strategy for decision makers or algorithm users by considering with their available computing resources.}




\maketitle

\section{Introduction}\label{sec1}

The Airplane Refueling Problem (ARP) was raised by Woeginger \cite{woeginger10} from a math puzzle problem \cite{puzzle58}. Suppose there are $n$ airplanes referred to $A_1, \cdots, A_n$, each $A_i$ can carry $v_i$ tanks of fuel, and consumes $c_i$ tanks of fuel per kilometer for $1\leqslant i \leqslant n$. The fleet starts to fly together to a same target at a same rate without getting fuel from outside, but each airplane can refuel to other remaining airplanes instantaneously during the trip and then be dropped out. The goal is to determine a drop out permutation $\pi = (\pi(1), \cdots, \pi(n))$ that maximize the flight distance by the last remaining airplane.

Previous research on ARP centralized on its complexity analysis and its exact algorithm \cite{hohn13,vasquez15,gamzu19,lijs19}. Related research also focused on equivalent problems of ARP such as the $n$-vehicle exploration problem \cite{lixy09,yu18,zhang21} and the single machine scheduling problem \cite{hohn13}. V{\'a}squez \cite{vasquez15} studied the structural properties of ARP such as connections between local precedence and global precedence. Iftah and Danny \cite{gamzu19} proposed a fast and easy-to-implement pseudo polynomial time algorithm that attained a Polynomial Time Approximation Scheme (PTAS) for ARP. We mainly focused on the performance of an algorithm on instances with large input size, and cared about that how long will a given algorithm take to solve real-life instances. We tried to find a superior algorithm that is able to handle the large size problem in any reasonable amount of time \cite{cook98}, and to explore the possibility of a fast exact algorithm which is efficient on ARP instances.

Li et al. \cite{lijs19} put forward a fast exact algorithm for the first time by searching for all the feasible solutions satisfied with some necessary conditions. They run the algorithm on some large scale of ARP instances and got relatively efficient results than previous algorithms. However, the authors did not investigate the theoretical computational complexity the fast exact algorithm will perform on worst case, and did not demonstrate that how fast and to what degree will the fast algorithm be when perform it on larger scale of ARP instances. Is it an exponential-time algorithm, or is it probably a polynomial-time algorithm? In this paper, we proposed the definition of the sequential feasible solution at first, and then we put forward the sequential search algorithm (SSA) by bubble sorting all of the sequential feasible solutions. The computational complexity of SSA depends on the number of sequential feasible solutions for any given ARP instance. Then we found that SSA has a characteristic that it will run in polynomial time on ARP instances when the input size is greater than an inflection point. According to the theoretical analysis on the algorithmic upper bound, we explained why there exists an inflection point for the algorithmic complexity. Moreover, we proved that the time complexity of SSA raises at a polynomial rate when the input size of $n$ is greater than an inflection point $N$.

At last we also proposed an efficient computability scheme to predict the sequential search algorithmic complexity for any given ARP instance. The idea of efficient computability is inspired by the following literature. In \cite{yu18}, the authors provided a conceptual mechanism of efficient computation on the $n$-vehicle exploration problem, which could acquire reasonable computing cost analysis to help decision makers choose corresponding algorithm. In book \cite{tardos06}, the authors mentioned that it is possible to quantify some precise senses in which an instance may be easier than the worst case, and to take advantage of these situations when they occur. They also pointed out that some "size" parameters has an enormous effect on the running time because an input with "special structure" can help us avoid many of the difficulties that can make the worst case intractable. In \cite{walteros20}, the authors provided a worst-case explanation for the phenomenon that why some real-life maximum clique instances are not intractable. They claimed that real-life instances of maximum clique problem often have a small clique-core gap, however real hard instances are still hard no matter the input size is large or small.

\section{Preliminary}
\label{sec:background}

We consider a permutation order $\pi$ and its related sequence $A_{\pi(1)} \Rightarrow A_{\pi(2)} \Rightarrow \cdots \Rightarrow A_{\pi(n)}$ (see \cref{fig:1}), where $A_{\pi(i)}$ refuels to $A_{\pi(j)}$ for any $i < j$. Let $S_{\pi}=\sum\limits_{i=1}^n{x_{\pi(i)}}$ denotes the total distance that the $n$ airplanes can approach, and $x_{\pi(i)}$ denotes the segmented part of distances that $A_{\pi(i)}$ travels farther than $A_{\pi(i-1)}$, which is also the flight distance that $A_{\pi(i)}$ contributes to $S_{\pi}$ separately. Then ARP is described as a scheduling problem, which aims to find a drop out permutation $\pi = (\pi(1), \pi(2), \cdots, \pi(n))$ that maximizes the following $S_{\pi}$. Since each $S_{\pi}$ corresponds to a permutation order of $n$ airplanes, there are totally $n!$ of feasible solutions need to be compared.

\begin{equation}
S_{\pi} = \frac{v_{\pi(1)}}{c_{\pi(1)}
\cdots + c_{\pi(n)}} + \cdots + \frac{v_{\pi(n)}}{c_{\pi(n)}}\label{eqt1}
\end{equation}

\begin{figure}[h]
\centering
\includegraphics[width=0.65\textwidth]{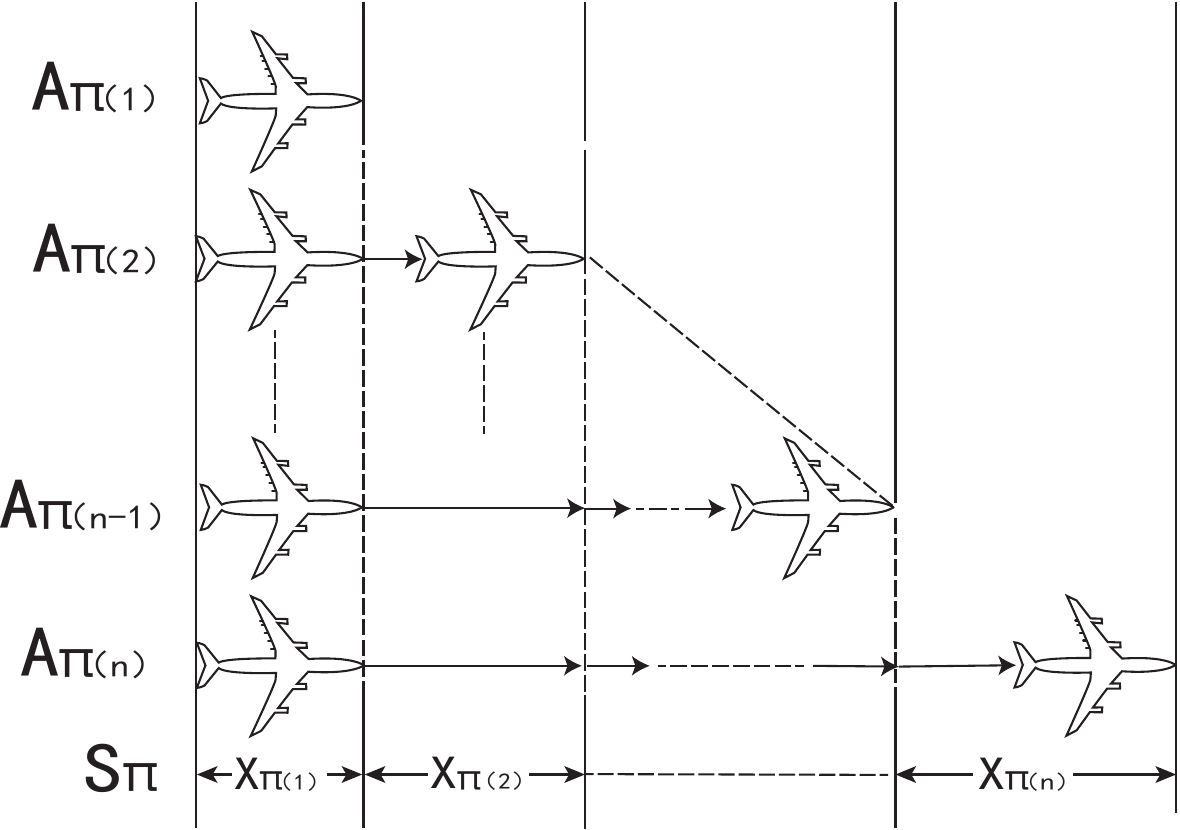}
\caption{Airplane refueling problem scheme.}\label{fig:1}
\end{figure}

Let $\mathcal{C}$ represents the cumulative sum of fuel consumption rates in an order $\pi = (\pi(1), \pi(2), \cdots, \pi(n))$, then $\mathcal{C}_{\pi(n)} = 0$, and $\mathcal{C}_{\pi(l)} = \sum\limits_{t=l+1}^nc_{\pi(t)}$ for any $l < n$. 

\section{Overall strategy}
\label{sec:org}

We would like to give the definition of polynomial-time algorithm \cite{garey79} before we proposed our algorithm.

\begin{definition}\label{def:2} A {\bf polynomial-time algorithm} is defined to be an algorithm whose time complexity function is $O(p(n))$ for some polynomial function $p(n)$ for all values of $n \geqslant 0$.\end{definition}

In \Cref{def:2}, the condition of $n \geqslant 0$ means that a polynomial-time algorithm is polynomial in time no matter the problem's input size is large or small. Normally, an algorithm runs in exponential time on small size of instances, it should be exponential-time on larger input size of instances. Thus, in studying the complexity of an algorithm, instances with larger inputs are of more priorities than other kinds because the larger inputs determine the applicability of the algorithm \cite{christosbook}.

Here, we mainly focused on seeking for a polynomial-time algorithm to solve ARP. We proposed SSA by numerating all of the sequential feasible solutions. The complexity of SSA depends on the number of sequential feasible solutions to the given ARP instance. We proved that for worst case of ARP instances, the number of its sequential feasible solutions is upper bounded by $2^{n-2}$. Then we focused on the following two major challenges. The first one is that will the running time of SSA decline down to polynomial-time when the input size of $n$ gets to be sufficiently large? The other challenge is that how can we predict the specific computational complexity of any given ARP instance in polynomial time before running SSA on it?

The answer to the first question is positive. For any ARP instance, if each airplane's single flight distance is limited in an interval $[L, S]$, then there must exist an constant $N$ such that SSA runs in polynomial time when the inputs size of $n$ is greater than an inflection point $N$ which does not depend on $n$. So we proved that SSA is a polynomial-time algorithm to solve ARP when its input size of $n$ is greater than an inflection point $N$.

\begin{definition}\label{def:3} When we run an algorithm on a problem, at first the time complexity grows in exponential time when the input size is small. But the time complexity of the algorithm will changes to grow in polynomial time when the input size is greater than $N$. For such situation, the point $N$ is defined as an {\bf inflection point}.\end{definition}

Besides, when $n$ is less than $N$, the number of sequential feasible solution is upper bounded by $2^{n-2}$, which is still less than $n^m$ for $m = N/2$. Therefore we proved that SSA is a polynomial-time algorithm to solve ARP.
%
%

The key point to answer the second question is to predict the inflection point $N$ for any given ARP instance in polynomial time. We proposed an algorithm to estimate the $N$ for worst case, and we improved the algorithm for general case by using heuristic method to attain the computational complexity in the design of the efficient computability scheme.

In \cref{sec:complexity}, we proposed the definition of sequential feasible solution, and constructed SSA by bubble sorting all of the sequential feasible solutions. We sharpened the upper bound of the number of sequential feasible solutions to $2^{n-2}$. In \cref{sec:palgorithm}, by exploring the computational complexity of ARP instances from a dynamic perspective, we found that the computational complexity of SSA grows at a slowing down rate when the input size of $n$ gets to be greater than an inflection point. In \cref{sec:efficom}, our efforts were devoted to construct the efficient computability scheme. We proposed a heuristic algorithm to estimate an inflection point $N = 2m$ for general case of ARP instances. Then we explained how to use the efficient computability scheme to predict computational complexity before we choose a proper algorithm for any given ARP instance.

\section{Sequential feasible solutions and the complexity upper bound}
\label{sec:complexity}

\subsection{Definition of sequential feasible solution}

A special case of $n$-vehicle exploration problem \cite{lixy09} is extended to ARP as follows.

\begin{theorem}\label{thm:1} In an ARP instance, if
\begin{displaymath}
     \frac{v_1}{c_1^2}\leqslant\frac{v_2}{c_2^2}\leqslant\cdots\leqslant\frac{v_{n-1}}{c_{n-1}^2}\leqslant\frac{v_n}{c_n^2}
\end{displaymath} and
\begin{displaymath}
\frac{v_1}{c_1}\leqslant\frac{v_2}{c_2}\leqslant\cdots\leqslant\frac{v_{n-1}}{c_{n-1}}\leqslant\frac{v_n}{c_n},
\end{displaymath}
then the optimal sequence is $A_1\Rightarrow
A_2\Rightarrow\cdots \Rightarrow A_n$.\end{theorem}

For general case of ARP instances, we randomly choose two neighbored airplanes $A_{i-1}$ and $A_{i}$ in $\pi = (1, \cdots, n)$.

Suppose $\mathcal{C}_i = \sum\limits_{t=i+1}^nc_{\pi(t)}$.

If $A_{i-1} \Rightarrow A_i$, then
\begin{displaymath}x_{\pi(i-1)} + x_{\pi(i)} = \frac{v_{i-1}}{c_{i-1} + c_i + \mathcal{C}_i} + \frac{v_i}{c_i + \mathcal{C}_i},\end{displaymath}

otherwise if $A_i \Rightarrow A_{i-1}$, then
\begin{displaymath}x_{\pi(i)} + x_{\pi(i-1)} = \frac{v_i}{c_i + c_{i-1} + \mathcal{C}_i} + \frac{v_{i-1}}{c_{i-1} + \mathcal{C}_i}.\end{displaymath}.

Consequently, if $\frac{v_i}{c_i \times (c_i + \mathcal{C}_i)} > \frac{v_{i-1}}{c_{i-1} \times (c_{i-1} + \mathcal{C}_i)}$, then $A_{i-1} \Rightarrow A_i$ is a better sequence; otherwise if $\frac{v_i}{c_i \times (c_i + \mathcal{C}_i)} \leqslant \frac{v_{i-1}}{c_{i-1} \times (c_{i-1} + \mathcal{C}_i)}$, then $A_i \Rightarrow A_{i-1}$ is a better sequence.

\begin{definition}\label{def:1} Given current cumulative sum of fuel consumption rates $C_o$, the {\bf relative distance factor} of $A_i$ is defined as $\varphi(A_i, \mathcal{C}_o) = \frac{v_i}{c_i \times (c_i + \mathcal{C}_o)}$. \end{definition}

The following \Cref{cor:1} is a corollary of \Cref{thm:1} by considering the definition of relative distance factor.

\begin{corollary}\label{cor:1} For an ARP instance with a given cumulative sum of fuel consumption rates $\mathcal{C}_o$, if
\begin{displaymath}
\frac{v_1}{c_1}\leqslant\frac{v_2}{c_2}\leqslant\cdots\leqslant\frac{v_{n-1}}{c_{n-1}}\leqslant\frac{v_n}{c_n}
\end{displaymath}
and
\begin{displaymath}
\frac{v_1}{c_1(c_1 + \mathcal{C}_o)}\leqslant\frac{v_2}{c_2(c_2 + \mathcal{C}_o)}\leqslant\cdots\leqslant\frac{v_{n-1}}{c_{n-1}(c_{n-1} + \mathcal{C}_o)}\leqslant\frac{v_n}{c_n(c_n + \mathcal{C}_o)},
\end{displaymath}
then the optimal sequence is $A_1\Rightarrow
A_2\Rightarrow\cdots \Rightarrow A_n$.\end{corollary}

\begin{definition}\label{def:4} A sequence $\pi = (\pi(1), \pi(2), \cdots, \pi(n))$ is called a {\bf sequential feasible solution}, if for each pair of airplanes $A_{\pi(i)}$ and $A_{\pi(j)}$ for $i < j$ in $\pi$, at least one of the following two inequalities holds.

\begin{equation}\label{eqt2}\frac{v_{\pi(i)}}{c_{\pi(i)}(c_{\pi(i)} + \mathcal{C}_{\pi(j)})} \leqslant \frac{v_{\pi(j)}}{c_{\pi(j)}(c_{\pi(j)} + \mathcal{C}_{\pi(j)})}\end{equation}

\begin{equation}\label{eqt3}\frac{v_{\pi(i)}}{c_{\pi(i)}(c_{\pi(i)} + \mathcal{C}_{\pi(i)})} \leqslant \frac{v_{\pi(j)}}{c_{\pi(j)}(c_{\pi(j)} + \mathcal{C}_{\pi(i)})} \end{equation}

Where $\mathcal{C}_{\pi(n)} = 0$, and $\mathcal{C}_{\pi(l)} = \sum\limits_{t=l+1}^nc_{\pi(t)}$ for any $l < n$.
\end{definition}

\begin{lemma}\label{lem:1}For each pair of airplanes $A_i$ and $A_j$ with $v_i /c_i^2 > v_j / c_j^2$ and $v_i / c_i < v_j / c_j$.
\begin{itemize}
\item[(1)] $v_i < v_j$, and $c_i < c_j$.

\item[(2)] For any $\mathcal{C}_o > 0$, $\frac{v_i}{c_i + \mathcal{C}_o}<\frac{v_j}{c_j + \mathcal{C}_o}$.\end{itemize}
\end{lemma}

\begin{proof}\begin{itemize}
\item[(1)] Since $\frac{v_i}{v_j} < \frac{c_i}{c_j} < \frac{v_i/c_i}{v_j/c_j} < 1$, it follows that $v_i < v_j$ and $c_i < c_j$.

\item[(2)] Since $c_i < c_j$, we have $\frac{v_i}{v_j} < \frac{c_i}{c_j} < \frac{c_i + \mathcal{C}_o}{c_j + \mathcal{C}_o}<1$, then $\frac{v_i}{c_i + \mathcal{C}_o} < \frac{v_j}{c_j + \mathcal{C}_o}$.\end{itemize}
Therefore for any $\mathcal{C}_o > 0$, $\frac{v_i}{c_i + \mathcal{C}_o}<\frac{v_j}{c_j + \mathcal{C}_o}$. \end{proof}

\begin{lemma}\label{lem:2}For each pair of airplanes $A_i$ and $A_j$ with $\frac{v_i}{c_i^2} > \frac{v_j}{c_j^2}$ and $\frac{v_i}{c_i} < \frac{v_j}{c_j}$. If there exists a sum of  fuel consumption rates $\mathcal{C}_o$, such that $\frac{v_i}{c_i(c_i + \mathcal{C}_o)} < \frac{v_j}{c_j(c_j + \mathcal{C}_o)}$. Then, for any $\mathcal{C} > \mathcal{C}_o$, it follows that $\frac{v_i}{c_i(c_i + \mathcal{C})} < \frac{v_j}{c_j(c_j + \mathcal{C})}$.\end{lemma}

\begin{proof}Since $\frac{v_ic_j}{v_jc_i} < \frac{c_i + \mathcal{C}_o}{c_j + \mathcal{C}_o} < \frac{c_i + C}{c_j + \mathcal{C}}$, it follows that $\frac{v_i}{c_i(c_i + C)} < \frac{v_j}{c_j(c_j + \mathcal{C})}$. \end{proof}

\begin{theorem}\label{thm:2} Given an ARP instance:
\begin{itemize}
\item[(1)] If it has feasible solutions, it must have sequential feasible solutions;

\item[(2)] The optimal feasible solution must be the optimal sequential feasible solution.\end{itemize}
\end{theorem}

\begin{proof}

\begin{itemize}
\item[(1)] Given $n$ airplanes $A_1, \cdots, A_n$, we can obtain a sequential feasible solution by running \Cref{alg:1}. The computational complexity of \Cref{alg:1} is $O(n^2)$.

\begin{algorithm}[htbp]
\caption{$Find-SequentialFS(\mathcal{A})$}
\label{alg:1}
\begin{algorithmic}[1]
\Require $\mathcal{A} = \{A_1, \cdots, A_n\}$
\Ensure $\pi = (\pi(1), \cdots, \pi(n))$
\State{$k \Leftarrow n$, $\mathcal{C}_o \Leftarrow 0$}
\While{$k \geqslant 1$}
    \State{$\varphi_o \Leftarrow \varphi(A_1, \mathcal{C}_o)$, see \Cref{def:1}}
    \State{$k_o \Leftarrow 1$}
    \For{$i := 2$ to $k$}
        \If{$\varphi(A_i, \mathcal{C}_o) > \varphi_o$}
            \State{$k_o \Leftarrow i$}
            \State{$\varphi_o \Leftarrow \varphi(A_i, \mathcal{C}_o)$}
        \EndIf
    \EndFor
    \State{$\mathcal{C}_o \Leftarrow \mathcal{C}_o + c_{k_o}$}
    \State{$\mathcal{A} \Leftarrow \mathcal{A} \backslash \{ A_{k_o} \}$}
    \If{$k < n$}
        \State{$count \Leftarrow 0$}
        \For{$i := k+1$ to $n$}
            \If{$k_o \geqslant \pi(i)$}
                \State{$count \Leftarrow count + 1$}
            \EndIf
        \EndFor
    \EndIf
    \State{$k_o \Leftarrow k_o + count$}
    \State{$\pi(k) \Leftarrow k_o$}
    \State{$k \Leftarrow k-1$}
\EndWhile
\end{algorithmic}
\end{algorithm}

For each pair of airplanes $A_{\pi(i)}$ and $A_{\pi(j)}$ for $i < j$ in $\pi$ obtained by \Cref{alg:1}, the \cref{eqt2} exists for all the time. Therefore, the $\pi$ is a sequential feasible solution.

\item[(2)] Given an $n$-airplane instance, we will prove that its optimal feasible solution must be a sequential feasible solution. Suppose the optimal feasible solution $\pi^*$ is not a sequential feasible solution, then there exist two airplanes $A_{\pi^*(i)}$ and $A_{\pi^*(j)}$ that do not satisfy the requirements of sequential feasible solution. Suppose $A_{\pi^*(i)} \Rightarrow A_{\pi^*(A)} \Rightarrow A_{\pi^*(j)}$, here $A_{\pi^*(A)}$ denotes the set of airplanes travel between $A_{\pi^*(i)}$ and $A_{\pi^*(j)}$. Suppose $C_{\pi^*(j)} = \sum\limits_{t=j+1}^nc_{\pi^*(t)}$ and $\mathcal{C}_{\pi^*(i)} = \sum\limits_{t=i+1}^nc_{\pi^*(t)}$. Then it follows that

    \begin{equation}\label{eqt4}
    \frac{v_{\pi^*(i)}}{c_{\pi^*(i)}(c_{\pi^*(i)} + \mathcal{C}_{\pi^*(j)})} > \frac{v_{\pi^*(j)}}{c_{\pi^*(j)}(c_{\pi^*(j)} + \mathcal{C}_{\pi^*(j)})}
    \end{equation}
     and
     \begin{equation}\label{eqt5}
    \frac{v_{\pi^*(i)}}{c_{\pi^*(i)}(c_{\pi^*(i)} + \mathcal{C}_{\pi^*(i)})} > \frac{v_{\pi^*(j)}}{c_{\pi^*(j)}(c_{\pi^*(j)} + \mathcal{C}_{\pi^*(i)})}.
    \end{equation}

    According to \Cref{thm:1}, suppose $\frac{v_{\pi^*(i)}}{c_{\pi^*(i)}^2} \leqslant \frac{v_{\pi^*(j)}}{c_{\pi^*(j)}^2}$. In this case, if $\frac{v_{\pi^*(i)}}{c_{\pi^*(i)}} \leqslant \frac{v_{\pi^*(j)}}{c_{\pi^*(j)}}$, then for any $\mathcal{C}_o > 0$,
    \begin{equation}\label{eqt6}
    \frac{v_{\pi^*(i)}}{c_{\pi^*(i)}(c_{\pi^*(i)} + \mathcal{C}_o)} \leqslant \frac{v_{\pi^*(j)}}{c_{\pi^*(j)}(c_{\pi^*(j)} + \mathcal{C}_o)}.
    \end{equation}

    Since \cref{eqt6} is contract with both \cref{eqt4} and \cref{eqt5}, when $\frac{v_{\pi^*(i)}}{c_{\pi^*(i)}^2} \leqslant \frac{v_{\pi^*(j)}}{c_{\pi^*(j)}^2}$, there is $\frac{v_{\pi^*(i)}}{c_{\pi^*(i)}} > \frac{v_{\pi^*(j)}}{c_{\pi^*(j)}}$.

    According to \Cref{lem:2}, for any $\mathcal{C} \geqslant \mathcal{C}_{\pi^*(j)}$, $A_{\pi^*(j)}$ must refuel to $A_{\pi^*(i)}$ in a sequential feasible solution, which is contrary to the assumption that $\pi^*$ is the optimal solution. If $\frac{v_{\pi^*(i)}}{c_{\pi^*(i)}^2} > \frac{v_{\pi^*(j)}}{c_{\pi^*(j)}^2}$, then it follows that $\frac{v_{\pi^*(i)}}{c_{\pi^*(i)}} < \frac{v_{\pi^*(j)}}{c_{\pi^*(j)}}$ and $c_{\pi^*(i)} < c_{\pi^*(j)}$. If we switch the order of $A_{\pi^*(i)}$ and $A_{\pi^*(j)}$ to attain a new ordering $\hat{\pi}$ and its related sequence is $A_{\pi^*(j)} \Rightarrow A_{\pi^*(A)} \Rightarrow A_{\pi^*(i)}$. It follows that
    \begin{displaymath}
    S_{\hat{\pi}} = \frac{v_{\pi^*(i)}}{c_{\pi^*(i)} + \mathcal{C}_{\pi^*(j)}} + \frac{v_{\pi^*(A)}}{c_{\pi^*(i)} + \mathcal{C}_{\pi^*(j)} + \mathcal{C}_{\pi^*(A)}} + \frac{v_{\pi^*(j)}}{c_{\pi^*(i)} + \mathcal{C}_{\pi^*(i)}}
    \end{displaymath}
    and
    \begin{displaymath}
    S_{\pi^*} = \frac{v_{\pi^*(j)}}{c_{\pi^*(j)} + \mathcal{C}_{\pi^*(j)}} + \frac{v_{\pi^*(A)}}{c_{\pi^*(j)} + \mathcal{C}_{\pi^*(j)} + \mathcal{C}_{\pi^*(A)}} + \frac{v_{\pi^*(i)}}{c_{\pi^*(j)} + \mathcal{C}_{\pi^*(i)}}.
    \end{displaymath}

    Since $c_{\pi^*(i)} < c_{\pi^*(j)}$, then it follows that $S_{\hat{\pi}} > S_{\pi^*}$, which is contrary to the assumption that $\pi^*$ is the optimal feasible solution.

    Therefore the optimal feasible solution of ARP must be the optimal sequential feasible solution. \end{itemize}

    \end{proof}

\subsection{Complexity analysis focused on worst case}
\begin{definition}\label{def:5} An ARP instance is called the {\bf worst case} when it has the greatest number of sequential feasible solutions for given input size of $n$.
\end{definition}

\begin{definition}\label{def:6} An airplane refueling instance is called a {\bf complete reverse order sequence}, if
 \begin{displaymath}\frac{v_1}{c_1^2}>\frac{v_2}{c_2^2}>\cdots>\frac{v_n}{c_n^2}\end{displaymath} and
\begin{displaymath}\frac{v_1}{c_1}<\frac{v_2}{c_2}<\cdots<\frac{v_n}{c_n}.\end{displaymath}\end{definition}


 In \cite{yu18}, the authors introduced cluster as a tool to model the computational complexity of $n$-vehicle exploration problem, and they claimed that the complete reverse order instances have greater computational complexity than the other kinds of instances. Similarly, we proposed the following \Cref{lem:3}.

\begin{lemma}\label{lem:3} The worst case must be a complete reverse order sequence for any given input size of $n \geqslant 2$.
\end{lemma}

\begin{proof}It is obvious that \Cref{lem:3} holds when $n = 2$.

Suppose a worst case of $k$-airplane instance is a complete reverse order sequence with $\frac{v_1}{c_1^2}>\frac{v_2}{c_2^2}>\cdots>\frac{v_k}{c_k^2}$ and $\frac{v_1}{c_1}<\frac{v_2}{c_2}<\cdots<\frac{v_k}{c_k}$, we will show that a worst case of $(k+1)$-airplane instance must be a complete reverse order sequence.

Suppose a worst case of $(k+1)$-airplane instance is not a complete reverse order sequence, which means that it doesn't satisfy with both $\frac{v_1}{c_1^2}>\frac{v_2}{c_2^2}>\cdots>\frac{v_k}{c_k^2} > \frac{v_{k+1}}{c_{k+1}^2}$ and $\frac{v_1}{c_1}<\frac{v_2}{c_2}<\cdots<\frac{v_k}{c_k} < \frac{v_{k+1}}{c_{k+1}}$. In this case, if we arrange any $A_i$ for $i \leqslant k$ as the farthest airplane, then the rest of $k$ airplanes can not form a complete reverse order sequence as supposed, so the rest of $k$ airplanes is not a worst case. Therefore increase $A_{k+1}$ into the $k$-airplane instance does not create the greatest increment of computational complexity, which is obviously contract to the assumption of worst case. \end{proof}

\begin{theorem}\label{thm:3} For any worst case of ARP instances with $\frac{v_1}{c_1^2}>\frac{v_2}{c_2^2}>\cdots>\frac{v_n}{c_n^2}$ and $\frac{v_1}{c_1}<\frac{v_2}{c_2}<\cdots<\frac{v_n}{c_n}$, let $Q_n$ represents the number of its sequential feasible solutions, then $Q_n \leqslant 2^{n-2}$ for $n \geqslant 2$.\end{theorem}

\begin{proof}Given a worst case of ARP with $\frac{v_1}{c_1^2}>\frac{v_2}{c_2^2}>\cdots>\frac{v_n}{c_n^2}$ and $\frac{v_1}{c_1}<\frac{v_2}{c_2}<\cdots<\frac{v_n}{c_n}$, we will prove that $Q_n$ is upper bounded by $2^{n-2}$. To reduce the search space to $2^{n-2}$ is not an obvious work, so we introduced combination into the proof because of $2^{n-2} = \binom{n-2}{0} + \binom{n-2}{1} + \cdots + \binom{n-2}{n-2}$. Here $\binom{n-2}{p}$ relates to the number of the sequential feasible solutions when $p$ airplanes are chosen that lie between $A_{\pi(n)}$ and $A_n$ for $0 \leqslant p \leqslant n-2$, here $\pi(n)$ is the farthest position in an order $\pi$.

When $n=2$, we've shown that $Q_2 = 1 = 2^{2-2}$. We use $\binom{0}{0}$ to calculate $Q_2$, which means, from combination point of view, that no airplane is chosen between  $A_{\pi(2)}$ and $A_2$, whose combination is account for increasing the number of sequential feasible solutions.

\begin{displaymath}\binom{0}{0}:\quad A_2 \Rightarrow A_1. \end{displaymath}

When $n=3$, we have $Q_3 \leqslant 2 = 2^{3-2}$. When $A_1$ is $A_{\pi(3)}$, $A_2$ is available between $A_1$ and $A_3$ that possible leads to $\binom{1}{1}$ sequential feasible solution. When $A_2$ is $A_{\pi(3)}$, there is another sequential feasible solution if and only if $\frac{v_3}{c_3 \times (c_3 + c_2)} > \frac{v_1}{c_1 \times (c_1 + c_2)}$. No other airplane is possible between $A_2$ and $A_3$, which leads to $\binom{1}{0}$ sequential feasible solution. There is no possibility when $A_3$ takes the farthest position in a sequential feasible solution. By running through all the possible situations, it follows that $Q_3 \leqslant\binom{1}{0} + \binom{1}{1}$.

\begin{equation}\label{eqt7}
\begin{aligned}
\binom{1}{1}:&\quad A_3 \Rightarrow A_2 \Rightarrow A_1, or \quad A_2 \Rightarrow A_3 \Rightarrow A_1\\
\binom{1}{0}:&\quad A_1 \Rightarrow A_3 \Rightarrow A_2\\
\end{aligned}\end{equation}

When $n=4$, let $A_1$ is $A_{\pi(4)}$, then the rest of $3$ airplanes correspond to at most $2$ sequential feasible solutions if it satisfies \cref{eqt4}. So the rest $3$ airplanes correspond to at most $\binom{1}{0} + \binom{1}{1} = 2$ sequential feasible solutions.

\begin{equation}\label{eqt8}\frac{v_4}{c_4(c_4+c_1+c_3)}>\frac{v_2}{c_2(c_2+c_1+c_3)}\end{equation}

If $A_2$ is $A_{\pi(4)}$, then it follows that:

\begin{equation}\label{eqt9}\frac{v_2}{c_2(c_2+c_2+c_3)}>\frac{v_1}{c_1(c_1+c_2+c_3)}.\end{equation}

Combining \cref{eqt8} with \cref{eqt9}, and by \Cref{lem:2}, it follows that:

\begin{equation}\label{eqt10}\frac{v_4}{c_4(c_4+c_2+c_3)}>\frac{v_2}{c_2(c_2+c_2+c_3)}>\frac{v_1}{c_1(c_1+c_2+c_3)}.\end{equation}

According to \cref{eqt10}, there is at most $\binom{1}{1} = 1$ sequential feasible solution as $A_1 \Rightarrow A_4 \Rightarrow A_3 \Rightarrow A_2$.

If $A_3$ is $A_{\pi(4)}$, there is at most $\binom{1}{0}$ possible sequential feasible solution.

By summing up all the above cases, we have $Q_4 \leqslant \binom{1}{0} + \binom{1}{1} + \binom{1}{1} + \binom{1}{0}$, therefore $Q_4 \leqslant \binom{2}{0} + \binom{2}{1} + \binom{2}{2}$.
\begin{equation}\label{eqt11}
\begin{aligned}
\binom{2}{2}:&\quad A_4 \Rightarrow A_3 \Rightarrow A_2 \Rightarrow A_1, or \quad A_3 \Rightarrow A_4 \Rightarrow A_2 \Rightarrow A_1\\
\binom{2}{1}:&\quad A_2 \Rightarrow A_4 \Rightarrow A_3 \Rightarrow A_1\\
 &\quad A_1 \Rightarrow A_4 \Rightarrow A_3 \Rightarrow A_2\\
\binom{2}{0}:&\quad A_1 \Rightarrow A_2 \Rightarrow A_4 \Rightarrow A_3\end{aligned}\end{equation}

When $n=k$, suppose there are at most $2^{k-2}$ sequential feasible solutions. Moreover, there are at most $\binom{k-2}{p}$ sequential feasible solutions when $p$ airplanes are chosen between $A_{\pi(k)}$ and $A_k$.

When $n=k+1$, suppose $\frac{v_k}{c_k^2} > \frac{v_{k+1}}{c_{k+1}^2}$ and $\frac{v_k}{c_k} < \frac{v_{k+1}}{c_{k+1}}$. Then $Q_{k+1}$ consists of $k$ parts.

\begin{itemize}
\item[ (1): ] There is at most $\binom{k-2}{0}$ sequential feasible solution when $A_k$ takes the farthest position and $A_{k+1}$ takes the second farthest position.
\item[] $\cdots$
\item[($p+1$):] When $p$ airplanes lie between $A_{\pi(k+1)}$ and $A_{k+1}$, the number of sequential feasible solutions consists of two parts. The first part refers to $\binom{k-2}{p-1}$ sequential feasible solutions when $A_1$ is $A_{\pi(k+1)}$ and the rest of $k$ airplanes form a $k$-airplane worst case; The other part refers to $\binom{k-2}{p}$ sequential feasible solutions when $A_1$ does not take the farthest position. Totally there are $\binom{k-2}{p-1} + \binom{k-2}{p} = \binom{k-1}{p}$ sequential feasible solutions in this case. Thus for each $\binom{k-2}{p}$, when we add $A_{k+1}$ into the new sequence, from the view of combination, the number of sequential feasible solution changes to be $\binom{k-2}{p} + \binom{k-2}{p-1} = \binom{k-1}{p}$.
\item[] $\cdots$
\item[ ($k$): ] There is at most $\binom{k-1}{k-1}$ sequential feasible solution when $(k-1)$ airplanes lie between $A_{\pi(k+1)}$ and $A_{k+1}$.
\end{itemize}

\begin{equation}\label{eqt12}
\begin{aligned}
Q_{k+1} \leqslant & \binom{k-1}{0} + (\binom{k-2}{0} + \binom{k-2}{1}) + \cdots + (\binom{k-2}{k-3} + \binom{k-2}{k-2}) + \binom{k-1}{k-1}\\
= & \binom{k-1}{0} + \binom{k-1}{1} + \cdots + \binom{k-1}{k-2} + \binom{k-1}{k-1}\\
= & 2^{k-1}\end{aligned}\end{equation}

Therefore, $Q_n$ is upper bounded by $2^{n-2}$. \end{proof}

\section{The SSA is a polynomial-time algorithm to solve ARP when its input size of $n$ is greater than an inflection point $N$}
\label{sec:palgorithm}

Given an ARP instance with $\frac{v_1}{c_1^2} > \cdots > \frac{v_n}{c_n^2}$ and $\frac{v_1}{c_1} < \cdots < \frac{v_n}{c_n}$. $\mathcal{C}_{i,j}$ is the minimal cumulative sum of fuel consumption rates that satisfies with $\varphi(A_j, \mathcal{C}) \geqslant \varphi(A_i, \mathcal{C})$, then $\mathcal{C}_{i,j}$ is formulated as follows.

\begin{equation}\label{eqt13}\mathcal{C}_{i,j} = \frac{v_{i}c_j^2 - v_jc_{i}^2}{v_jc_{i} - v_{i}c_j}, 1 \leqslant i < j \leqslant n \end{equation}

For any $\mathcal{C} \geqslant \mathcal{C}_{i,j}$, it always follows that $\varphi(A_j, \mathcal{C}) \geqslant \varphi(A_{i}, \mathcal{C})$.

\begin{lemma}\label{lem:4} For a worst case of ARP instances with $\frac{v_1}{c_1^2} > \cdots > \frac{v_n}{c_n^2}$ and $L = \frac{v_1}{c_1} < \cdots < \frac{v_n}{c_n} = S$. There exists an index number $m$ associated with a cumulative sum of fuel consumption rates as $\mathcal{C}_o$, such that $\varphi(A_n, \mathcal{C}) > \varphi(A_{i}, \mathcal{C})$ for any $\mathcal{C} > \mathcal{C}_o$ and $1 \leqslant i \leqslant n-1$.\end{lemma}

\begin{proof} Let $\mathcal{C}_o$ equals to the maximal $\mathcal{C}_{i,n}$ for all $1 \leqslant i \leqslant n-1$.

 \begin{equation}\label{eqt14}\frac{v_n}{c_n(c_n + \mathcal{C}_o)} > \frac{v_i}{c_i(c_i + \mathcal{C}_o)}, 1 \leqslant i \leqslant n-1\end{equation}

Since $c_1 < \cdots < c_n$, we could determine $m$ by iteratively adding $c_i$ to $\mathcal{C}_m = \sum\limits_{k=1}^mc_k$ from $i = 1$ to $i = n$, until $\mathcal{C}_m$ is no longer less than $\mathcal{C}_o$. To be of more generality, we can also determine an upper bound of $m$ as $\bar{m} = \mathcal{C}_o / c_1$. 

The existence of $m$ is obvious, and we could prove it by contradiction.

If $m$ doesn't exist, there exist an airplane $A_j$ that can not meet \cref{eqt14}, which means that $\mathcal{C}_{j, n}$ has no upper bound.

\begin{equation}\label{eqt15}\mathcal{C}_{j, n} = \frac{v_{j}c_n^2 - v_nc_{j}^2}{v_nc_{j} - v_{j}c_n}\end{equation}

If $\mathcal{C}_{j, n} \rightarrow +\infty$, then $v_n / c_n \rightarrow +\infty$, which leads to contradiction with the assumption of $v_n / c_n \leqslant S$.

Given the index number $m$ and its related $\mathcal{C}_o$, if $\varphi(A_n, \mathcal{C}_o) > \varphi(A_{i}, \mathcal{C}_o)$, then $\varphi(A_n, \mathcal{C}) > \varphi(A_{i}, \mathcal{C})$ for any $\mathcal{C} > \mathcal{C}_o$ and $1 \leqslant i \leqslant n-1$. \end{proof}

\begin{theorem}\label{thm:4} For a worst case of ARP instance with $\frac{v_1}{c_1^2} > \cdots > \frac{v_n}{c_n^2}$ and $L = \frac{v_1}{c_1} < \cdots < \frac{v_n}{c_n} = S$. If
\begin{displaymath}
     \frac{v_n}{c_n} > \frac{v_i}{c_i} > \frac{v_j}{c_j}
\end{displaymath} and
\begin{displaymath}
\frac{v_n}{c_n^2} < \frac{v_i}{c_i^2} < \frac{v_j}{c_j^2}.
\end{displaymath}
Suppose we have determined a cumulative sum of fuel consumption rates $\mathcal{C}_o$ and an index number $m$ according to \Cref{lem:4}. If \begin{displaymath}
\varphi(A_n, \mathcal{C}_o) > \varphi(A_i, \mathcal{C}_o)\ and\ \varphi(A_n, \mathcal{C}_o) > \varphi(A_j, \mathcal{C}_o),
\end{displaymath}
then
\begin{displaymath}
     \varphi(A_i, \mathcal{C}_o) > \varphi(A_j, \mathcal{C}_o).
\end{displaymath}
\end{theorem}
\begin{proof} Let $S_i = a_i/b_i$ according to $[L, S_i]$, and $S_j = a_j/b_j$ according to $[L, S_j]$, then $[L, S_j] \subset [L, S_i] \subset [L, S]$.

Without loss of generality, suppose no $\mathcal{C}_j$ exists such that
\begin{equation}\label{eqt29}
     \varphi(A_j, \mathcal{C}_j) > \varphi(A_{j-1}, \mathcal{C}_j) > \cdots > \varphi(A_1, \mathcal{C}_j).
\end{equation}

If
\begin{equation}\label{eqt30}
     \varphi(A_j, \mathcal{C}_j) > \varphi(A_i, \mathcal{C}_j),
\end{equation}

Then we can determine $\mathcal{C}_i$ by iteratively adding $b_{j+1}, b_{j+2}, \cdots$ to $\mathcal{C}_j$ such that
\begin{equation}\label{eqt31}
     \varphi(A_i, \mathcal{C}_i) > \varphi(A_j, \mathcal{C}_i).
\end{equation}

Therefore, we have $\mathcal{C}_i \geqslant \mathcal{C}_j$.

If \cref{eqt30} dose not exist, then
\begin{displaymath}
     \varphi(A_i, \mathcal{C}_j) > \varphi(A_j, \mathcal{C}_j),\ and\ \mathcal{C}_j = \mathcal{C}_i.
\end{displaymath}

By conclusion, we have $\mathcal{C}_o \geqslant \mathcal{C}_{n-1} \geqslant \cdots \mathcal{C}_{n-m}$.
\end{proof}

\begin{lemma}\label{lem:5} For a worst case of ARP instances with $\frac{v_1}{c_1^2} > \cdots > \frac{v_n}{c_n^2}$ and $L = \frac{v_1}{c_1} < \cdots < \frac{v_n}{c_n} = S$. If we have determined an index number $m$, then $Q_n < \frac{m^2}{n}\binom{n}{m}$ when $n$ is greater than $2m$.\end{lemma}

\begin{proof} In the proof of \Cref{thm:3} we introduced combination number to calculate the number of sequential feasible solutions. For a given worst case of ARP instance, $Q_n$ is upper bounded by $2^{n-2}$, which is equal to $\sum_{p=0}^{n-2}\binom{n-2}{p}$. Hence $Q_n$ is divided into $n-1$ parts, and each part has at most $\binom{n-2}{k}$ sequential feasible solutions for $0 \leqslant p \leqslant n-2$. Here, $\binom{n-2}{p}$ means that only $p$ airplanes are chosen from $n-2$ airplanes to travel between $A_{\pi(n)}$ and $A_n$ in a sequential feasible solution. Once we have determined all of the airplanes that will take farther than $A_n$, then there is only one dropout order for the rest of airplanes that in a sequential feasible solution according to \Cref{thm:4}. Thus there are $p+1$ airplanes have the possibility to take precedence over $A_n$ and to increase the number of sequential feasible solutions.

According to \Cref{lem:4}, given a worst case of ARP instance, we can determine an index number $m$, such that for a worst case of ARP instances with $n > 2m$, it follows that at most $m-1$ airplanes are available to be chosen to locate between $A_{\pi(n)}$ and $A_n$, which means that $Q_n$ is composed by at most $m$ parts such as $\binom{n-2}{0}$, $\binom{n-2}{1}$, $\cdots$, $\binom{n-2}{m-1}$.

\begin{equation}\label{eqt16}
Q_n = \binom{n-2}{0} + \binom{n-2}{1} + \cdots + \binom{n-2}{m-1}\end{equation}

We will show that $Q_n$ in \cref{eqt16} is less than $\frac{m^2}{n}\binom{n}{m}$ when $n$ is greater than $2m$.

\begin{equation}\label{eqt17} \binom{n-2}{0} < \binom{n-2}{1} < \cdots < \binom{n-2}{m-1}\end{equation}

\begin{equation}\label{eqt18}
\begin{aligned}
Q_n &< m\binom{n-2}{m-1} \\
          &= m \times \frac{(n-2) \times (n-3) \times \cdots \times (n-m)}{(m-1) \times (m-2) \times \cdots \times 1}  \\
          &= \frac{m^2 \times (n-m)}{n \times (n-1)} \times \frac{n \times (n-1) \times (n-2) \times \cdots \times (n-m+1)}{m \times (m-1) \times (m-2) \times \cdots \times 1}  \\
          &< \frac{m^2}{n} \times \frac{n \times (n-1) \times (n-2) \times \cdots \times (n-m+1)}{m \times (m-1) \times \cdots \times 1}  \\
          &= \frac{m^2}{n} \times \binom{n}{m}
\end{aligned}
\end{equation}

Combined \cref{eqt17} with \cref{eqt18}, we have $Q_n < \frac{m^2}{n}\binom{n}{m}$ when $n > 2m$. \end{proof}

\begin{theorem}\label{thm:5} For a worst case of ARP instances with $\frac{v_1}{c_1^2} > \cdots > \frac{v_n}{c_n^2}$ and $L = \frac{v_1}{c_1} < \cdots < \frac{v_n}{c_n} = S$. Suppose we have determined an index number $m$, then there exists an inflection point $N = 2m$ such that $Q_n < \frac{m^2}{n}\binom{n}{m}$ for all values of $n > N$. In addition, $Q_n < n^{\underline{m}}$ when $n > N$.\end{theorem}

\begin{proof} According to \Cref{lem:5}, $Q_n$ is less than $\frac{m^2}{n}\binom{n}{m}$ when $n > N$.

\begin{equation}\label{eqt19}
Q_n < \frac{m^2}{n}\binom{n}{m} = \frac{m^2 \times n \times(n-1) \times \cdots \times(n-m+1)}{n \times m \times (m-1) \times \cdots \times 1} = n^{\underline{m}} < n^m
\end{equation}

Thus $Q_n$ is less than $n^{\underline{m}}$ when $n > N$.

\end{proof}

\begin{algorithm}[htbp]
\caption{$Sequential-Search-Algorithm(\mathcal{A})$} 
\label{alg:2}
\begin{algorithmic}[1]
\Require $\mathcal{A}$
\Ensure $\pi^*$ and $S^*$
\State{{\bf Step 1:} Run $Search-all-Sequential(\mathcal{A}, 0, 0)$}
  \State{Define Function: $\Pi \Leftarrow Search-all-Sequential(\mathcal{A}, \varphi_o, \mathcal{C}_o)$}
  \State{Sort $A_1, A_2, \cdots, A_n$ in decreasing order of $\varphi(A_i, \mathcal{C}_o)$ (see \Cref{def:1}) to get a new sequence $sa(1), \cdots, sa(n)$}
  \If{$n=1$}
      \If{$\varphi(A_{sa(1)}, \mathcal{C}_o) < \varphi_o$ or $\varphi_o = 0$}
          \State{$\pi \Leftarrow (sa(1))$}
          \State{$\Pi \Leftarrow \{\pi\}$}
      \EndIf
  \EndIf
  \If{$n=2$}
      \If{$\varphi(A_{sa(1)}, \mathcal{C}_o) < \varphi_o$ or $\varphi_o = 0$}
          \State{$\pi \Leftarrow ({sa(2)}, {sa(1)})$}
          \State{$\Pi \Leftarrow \{\pi\}$}
      \EndIf
  \EndIf
  \If{$n>2$}
      \For{$i := 1$ to $n-1$}
          \If{$\varphi(A_{sa(i)}, \mathcal{C}_o) < \varphi_o$ or $\varphi_o = 0$}
              \State{Let $A_{sa(i)}$ run the farthest position in the current sequence}
              \For{$j := i+1$ to $n$}
                 \If{$\varphi(A_{sa(i)}, \mathcal{C}_o) \geqslant \varphi(A_{sa(j)}, \mathcal{C}_o)$}
                     \State{$\mathcal{A} \Leftarrow \mathcal{A} \backslash \{A_{\pi({sa(i)})}\}$}
                     \State{$\varphi_o \Leftarrow \varphi(A_{sa(i)}, \mathcal{C}_o + c_{sa(i)})$}
                     \State{$\Pi^j \Leftarrow Search-all-sequential(\mathcal{A}, \varphi_o, \mathcal{C}_o + c_{sa(i)})$}
                     \For{Each element $k$ in $\Pi^j$}
                      \If{$k \geqslant sa(i)$}
                         \State{$k \Leftarrow k+1$}
                      \EndIf
                    \EndFor
                    \State{$\Pi \Leftarrow \{(\Pi^j, sa(i))\}$}
                 \EndIf
              \EndFor
         \EndIf
      \EndFor
  \EndIf
\State{{\bf Step 2:} Run $Bubble-sorting-Sequential(\mathcal{A}, \Pi)$}
  \State{Define Function: $(\pi^*, S^*) \leftarrow Bubble-sorting-Sequential(\mathcal{A}, \Pi)$}
  \State{$m \Leftarrow$ the number of rows in $\Pi$}
  \State{$S^* \Leftarrow 0$}
  \For{$i := 1$ to $m$}
    \State{$S^o \Leftarrow S_{\Pi(i)}$, see \cref{eqt1}}
    \If{$S^o > S^*$}
        \State{$S^* \Leftarrow S_{\Pi(i)}$}
        \State{$\pi^* \Leftarrow \Pi(i)$}
    \EndIf
  \EndFor
\end{algorithmic}
\end{algorithm}

Similar with the fast exact algorithm in \cite{lijs19}, SSA was proposed in \Cref{alg:2} by bubble sorting all of the sequential feasible solutions to get access to the optimal feasible solution. The computational complexity of \Cref{alg:2} is $O(n^2Q_n^2)$. According to \Cref{thm:3}, $Q_n \leqslant 2^{n-2}$ for ARP instances with small input size. However, by \Cref{thm:5}, $Q_n < n^m$ for ARP instances with input size of $n$ is greater than an inflection point $N = 2m$.

%

\begin{corollary}\label{cor:3}For a worst case of APR instance with $\frac{v_1}{c_1^2} > \cdots > \frac{v_n}{c_n^2}$ and $L = \frac{v_1}{c_1} < \cdots < \frac{v_n}{c_n} = S$. There exists an inflection point $N$, such that SSA runs in polynomial time when the input size of $n$ is greater than $N$.\end{corollary}

The computational complexity of SSA is presented as follows.

\begin{equation}
\label{eqt20}
O(n^2Q_n^2) = \begin{cases}
\begin{aligned}
O(n^22^{2n-4}), &  &2 \leqslant n \leqslant 2m \\
O(n^{2m+2}), &  &n > 2m
\end{aligned}
\end{cases}
\end{equation}

Given a large set of airplanes with an inflection point $N$, a turn toward polynomial running time of SSA occurs when the input size of $n$ is greater than $N$. The turning point of $N$ is regarded as an inflection point: when $n$ is less than $2m$, the upper bound of $Q_n$ has risen sharply but when $n$ ia greater than $2m$, although the amount of possible sequential feasible solutions is still rising, its growth rate is slowing down.



\begin{theorem}\label{thm:6} Suppose $\mathcal{A}$ is a set of airplanes, the single flight distance of each $A_i \in \mathcal{A}$ is limited by $L \leqslant v_i / c_i \leqslant S$. For any ARP instance chosen from $\mathcal{A}$, there must exist an inflection point $\bar{N}$, such that when the input size of $n$ is greater than $\bar{N}$, the time complexity of SSA running on the ARP instance changes to grow at polynomial rate. Besides, the $\bar{N}$ - obtained from the worst case of $\mathcal{A}$ - only depends on the interval $[L, S]$ and does not depend on the input size of $n$.\end{theorem}

\begin{proof} In proof of \Cref{lem:4}, we have discussed the idea and method to prove the existence of inflection point $N$ for a given ARP instance.

The upper bound of $N$ for airplanes' set $\mathcal{A}$ is determined by the following process.

Given the worst case of ARP instances of $\mathcal{A}$ with $\frac{v_1}{c_1^2} > \cdots > \frac{v_n}{c_n^2}$ and $L = \frac{v_1}{c_1} < \cdots < \frac{v_n}{c_n} = S$. Assuming $\mathcal{C}_o$ equals to the maximal $\mathcal{C}_{i,n}$ (see \cref{eqt13}) for $1 \leqslant i < n$ corresponding to the worst case. Let $\bar{m} = \mathcal{C}_o / c_1$. Since $c_1 \leqslant \cdots \leqslant c_n$, then an inflection point $\bar{N} = 2\bar{m}$ must be greater than any inflection point of ARP instances in $\mathcal{A}$.

Next we will prove the existence of $\bar{N}$ by contradiction.

For an ARP instance chosen from airplanes' set $\mathcal{A}$, its related index number $\hat{m}$ is less than $\hat{\mathcal{C}} / c_1$. Here, $\hat{\mathcal{C}}$ is the maximal $\mathcal{C}_{i,j}$ as defined in \cref{eqt13}. If the inflection point $\hat{N} = 2\hat{m}$ is greater than $\bar{N}$, then for each $A_i$ in the ARP instance, \cref{eqt21} always holds.

\begin{equation}\label{eqt21}\frac{v_n}{c_n(c_n + \hat{\mathcal{C}})} > \frac{v_i}{c_i(c_i + \hat{\mathcal{C}})}, i = 1, \ldots, n-1\end{equation}

Since $\bar{N}$ is obtained from the worst case of $\mathcal{A}$, \cref{eqt22} always holds.
\begin{equation}\label{eqt22}\frac{v_n}{c_n(c_n + \mathcal{C}_o)} > \frac{v_i}{c_i(c_i + \mathcal{C}_o)}, i = 1, \ldots, n-1\end{equation}

According to \Cref{lem:2} and the assumption of the worst case, we have $\mathcal{C}_o \geqslant \hat{\mathcal{C}}$ and $\bar{m} \geqslant \hat{m}$. So there is an inflection point as $\bar{N} = 2\bar{m}$, such that for any ARP instance chosen from $\mathcal{A}$, SSA is a polynomial-time algorithm to solve the ARP instance when its input size of $n$ is greater than the inflection point $\bar{N}$.

Next we will prove that $\bar{N}$ dose not depend on the input size of $n$.

For the worst case of ARP instances in $\mathcal{A}$, according to \Cref{thm:4}, it follows that $\mathcal{C}_o \geqslant \mathcal{C}_{n-1,n} \geqslant \cdots \geqslant \mathcal{C}_{1,n}$, which refers to $[L, S_1] \subset [L, S_2] \subset \cdots \subset [L, S]$.

\begin{equation}\label{eqt23}
\begin{aligned}
\mathcal{C}_o & = \mathcal{C}_{n-1,n} \\
          &= \frac{v_{n-1}c_n^2 - v_nc_{n-1}^2}{v_nc_{n-1} - v_{n-1}c_n} \\
\end{aligned}
\end{equation}

\begin{equation}\label{eqt24}
\bar{m} = \frac{\mathcal{C}_o}{c_1} = \frac{\frac{v_{n-1}}{c_{n-1}}(\frac{c_n}{c_{1}} - \frac{c_{n-1}}{c_{1}})}{\frac{v_n}{c_n} - \frac{v_{n-1}}{c_{n-1}}} - \frac{c_{n-1}}{c_{1}}\\
\end{equation}

In practice, $\frac{v_n}{c_n} - \frac{v_{n-1}}{c_{n-1}}$ can not be infinitely small, otherwise the related two airplanes are regarded to be identity.

So the upper bound of inflection point $\bar{N} = 2\bar{m}$ in \cref{eqt24} is just decided by the values of $L$ and $S$. It does not depend on how many airplanes located between $[L, S]$ and it dose not depend on the input size of $n$.

\end{proof}


According to \Cref{thm:6}, suppose $\mathcal{A}$ is a set of $n$ airplanes with $L = \frac{v_1}{c_1} < \cdots < \frac{v_n}{c_n} = S$. For each ARP instance in $\mathcal{A}$, there exists an index number $m$, such that the number of sequential feasible solutions of the instance is less than $n^m$ when $n > 2m$.

\begin{equation}
\label{eqt25}
\begin{cases}
\begin{aligned}
Q_n &\leqslant 2^{n-2}, &  &2 \leqslant n \leqslant 2m \\
Q_n &< n^{m}, &  &n > 2m
\end{aligned}
\end{cases}
\end{equation}

Suppose an ARP instance chosen from $\mathcal{A}$ is limited in an interval $[L, S_1]$, $S_1 < S$. The related index number of this ARP instance is $m_1$ and the related $Q_n$ must be less than $n^{m_1}$ when $n > 2m_1$. Moreover, assuming $\bar{m}$ is the upper bound of $m$ for $\mathcal{A}$ with an interval $[L, S]$, according to \Cref{thm:6}, for any ARP instance chosen from $\mathcal{A}$, $Q_n$ is displayed in \cref{eqt26}. Since $S_1 < S$, if follows that $m_1 < \bar{m}$.
%

\begin{equation}
\label{eqt26}
\begin{cases}
\begin{aligned}
Q_n &\leqslant 2^{n-2}, &  &2 \leqslant n \leqslant 2\bar{m} \\
Q_n &< n^{\bar{m}}, &  &n > 2\bar{m}
\end{aligned}
\end{cases}
\end{equation}


\begin{corollary}\label{cor:4} Suppose $\mathcal{A}$ is a set of airplanes, the single flight distance of each $A_i \in \mathcal{A}$ is limited by $L \leqslant v_i / c_i \leqslant S$. Then \Cref{alg:2} is a polynomial-time algorithm for any ARP instance chosen from $\mathcal{A}$.\end{corollary}

\begin{proof} According to \cref{eqt25}, SSA is a polynomial-time algorithm to solve ARP instance when its input size of $n$ is greater than an inflection point $N = 2m$.


\begin{equation} \label{eqt27} Q_n = 2^{n-2} < n^m,  2 \leqslant n \leqslant 3 \end{equation}

\begin{equation}
\label{eqt28}
\begin{aligned}
Q_n &\leqslant 2^{n-2}, &4 \leqslant n \leqslant 2m \\
    &\leqslant 2^{2m - 2} &\\
    &\leqslant 4^{m-1} &\\
    &< n^m &
\end{aligned}
\end{equation}

\end{proof}

The following \Cref{alg:3} is proposed to estimate the $m$ and $Q_n$.

\begin{algorithm}[htbp]
\caption{$Estimate-SSA-Complexity(\mathcal{A})$} 
\label{alg:3}
\begin{algorithmic}[1]
\Require $\mathcal{A}$
\Ensure $m$, and $Q_n$
\State{Sort $\mathcal{A}$ in decreasing order of $\varphi(A_i, 0)$, see \Cref{def:1}}
\State{$m \Leftarrow 1$}
\State{$\mathcal{C}_m \Leftarrow c_1$}
\State{$\mathcal{A} \Leftarrow \{A_2, \cdots, A_{n-1}\}$}
\While{$m<n-1$ and $\varphi(A_n, \mathcal{C}_m) < max\{\varphi(A_i, \mathcal{C}_m)$ and $A_i\in \mathcal{A}\}$}
    \State{$m \Leftarrow m + 1$}
    \State{$\mathcal{C}_m \Leftarrow \mathcal{C}_m + c_m$}
    \State{$\mathcal{A} \Leftarrow \mathcal{A} \backslash \{A_m, \cdots, A_{n-1}\}$}
    \State{$Q_n \Leftarrow \frac{m^2}{n}\binom{n}{m}$}
\EndWhile
	\end{algorithmic}
\end{algorithm}

For any ARP instance in $\mathcal{A}$, especially for those worst cases with different intervals of single flight distance, a phenomenon of inflection point was discovered by observing the change in the number of sequential feasible solutions. The dependency relationship between inflection point and interval of single flight distance was found. A method for finding a constant inflection point was presented. Thus it was proved that SSA is a polynomial-time algorithm to solve ARP.

\section{Efficient computability scheme for ARP}
\label{sec:efficom}

Given an instance of ARP, we could acknowledge within polynomial time how much time SSA runs on it, according to which we choose a proper algorithm considering available running time. In \Cref{thm:3}, we claimed that the upper bound of $Q_n$ is $2^{n-2}$. However, in \Cref{thm:5} we proved that an instance is polynomial-time solvable when its input size of $n$ is greater than $N$. Here $N$ is regarded as an inflection point of the number of possible sequential feasible solutions: the computational complexity of given instance grows exponentially for $n \leqslant N$, but grows polynomially for $n > N$. For this reason, to predict $Q_n$, we shall just find the inflection point $N$ and its related index number $m$.

The main idea of efficient computability is that, given an ARP instance, we shall forecast in polynomial time the particular computational complexity that SSA runs on it, which provides useful information to decision makers or algorithm users before they try to solve it.

According to \cref{eqt16} in \cref{sec:palgorithm}, $Q_n$ is an aggregation of $m$ parts, and each part is upper bounded by $\binom{n-2}{i-1}$ for $1 \leqslant i \leqslant m$. Here $\binom{n-2}{i-1}$ is the maximal amount of potential sequential feasible solutions when there are $i$ airplanes take precedence over the airplane with the greatest $v_j / c_j$ for $1 \leqslant j \leqslant n$. We propose the following \Cref{alg:4} to get an approximate value of $\widetilde{m}$ and $\widetilde{Q_n}$ by using heuristic method.

\begin{algorithm}[htbp]
\caption{$Heuristic-SSA-Complexity(\mathcal{A})$} 
\label{alg:4}
\begin{algorithmic}[1]
\Require $\mathcal{A}$
\Ensure $\widetilde{m}$ and $\widetilde{Q_n}$
\State{Sort $\mathcal{A}$ in decreasing order of $\varphi(A_i, 0)$, see \Cref{def:1}}
\For{$i := 1$ to $n-1$}
    \State{Let $A_i$ takes the farthest position}
    \State{Running \Cref{alg:1} to get a sequence $\pi^i$}
    \State{Calculate $S_{\pi^i}$ according to \cref{eqt1}}
\EndFor
\State{$\pi^* \Leftarrow arg\ max\{S_{\pi^i}, 1 \leqslant i < n\}$}
\State{Find $\widetilde{m}$ such that $\pi^*(\widetilde{m}+1) = arg\ max\{\frac{v_{\pi^*(i)}}{c_{\pi^*(i)}}\}$}
\State{$\widetilde{Q_n} \Leftarrow \frac{\widetilde{m}^2}{n}\binom{n}{\widetilde{m}}$}
\end{algorithmic}
\end{algorithm}

We consider some possible situations in addition to pose a comprehensive mechanism of efficient computability:
\begin{itemize}
\item[(1)] Special case in \Cref{thm:1} with $\frac{v_1}{c_1^2} \leqslant \cdots \leqslant \frac{v_n}{c_n^2}$, and $\frac{v_1}{c_1} \leqslant \cdots \leqslant \frac{v_n}{c_n}$. Such instances are easy to solve.
\item[(2)] Not complete reverse order sequence. Such instances can be transformed into lower dimensional form which are relatively easy to solve.
\item[(3)] Complete reverse order sequence including the following three categories:
\begin{itemize}
\item[(3-1)] When $m<<n$, such instances are easy to solve;
\item[(3-2)] When $m \approx n/2$, such instances are similar with a worst case. But once $m$ is fixed, and when $n$ gets sufficiently large, the number of sequential feasible solutions changes to be less than $\frac{m^2}{n}\binom{n}{m}$;
\item[(3-3)] When $m \approx n$, such instances are easy to solve at this case. But if we increase the input scale of $n$ to $2m$, it changes to be case (3-2) at first and to be (3-1) at last.\end{itemize}\end{itemize}

\section{Numerical illustration}
\label{sec:example}

\subsection{Example $1$}

We build an ARP instance contains $60$ airplanes with $\frac{v_1}{c_1^2} > \cdots > \frac{v_n}{c_n^2}$ and $\frac{v_1}{c_1} < \cdots < \frac{v_n}{c_n}$. The data of example $1$ is partly displayed in \Cref{tab:3}.

\begin{table}[htbp]
\caption{Data of Example $1$}\label{tab:3} %
\begin{tabular}{@{}lllllllllllll@{}}
\toprule
$\mathcal{A}$ & $A_1$ & $A_2$ & $A_3$ & $A_4$ & $A_5$ & $A_6$ & $A_{10}$ & $A_{20}$ & $A_{30}$ & $A_{40}$ & $A_{50}$ & $A_{60}$\\
\midrule
  $v_i$  & $4$ & $6.27$ & $8.73$ & $11.36$ & $14.18$ & $17.18$ & $31$ & $78.27$ & $143.73$ & $227.36$ & $329.18$ & $449.18$\\
    $c_i$  & $2$ & $3$ & $4$ & $5$ & $6$ & $7$ & $11$ & $21$ & $31$ & $41$ & $51$ & $61$\\
\botrule
\end{tabular}
\end{table}


The index number of Example $1$ is $m = 6$ by running \Cref{alg:3}. Suppose we choose different fleet of airplanes from the $60$ airplanes in example $1$ to get various ARP instances with different input size of $n$, we can calculate their number of sequential feasible solutions of $Q_n$ by running part $1$ in \Cref{alg:2}. The different $Q_n$ is displayed in \Cref{tab:4}.

\begin{table}[htbp]
\caption{Time complexity results of Example $1$}\label{tab:4} %
\begin{tabular}{@{}lllllllllllll@{}}
\toprule
$n$ & $6$ & $12$ & $24$ & $36$ & $48$ & $60$\\
  \midrule
   $2^{n-2}$  & $16$ & $1 024$ & $4.19 \times 10^{6}$ & $1.72 \times 10^{10}$ & $7.04 \times 10^{13}$ & $2.88 \times 10^{17}$\\
     &     &         &                      &                      &                      &           \\
   $\frac{m^2}{n}\binom{n}{m}$  & $6$ & $2,772$ & $2.02 \times 10^{5}$ & $1.95 \times 10^{6}$ & $9.20 \times 10^{6}$ & $3.00 \times 10^{7}$\\
    \midrule
  $Q_n$  & $1$ & $399$ & $2,323$ & $4,374$ & $6,355$ & $8,406$\\
\botrule
\end{tabular}
\end{table}

\subsection{Example $2$}

We build an ARP instance contains $1000$ airplanes with $\frac{v_1}{c_1^2} > \cdots > \frac{v_n}{c_n^2}$ and $\frac{v_1}{c_1} < \cdots < \frac{v_n}{c_n}$. The data of example $2$ is partly displayed in \Cref{tab:5}.

\begin{table}[htbp]
\caption{Data of Example $2$}\label{tab:5} %
\begin{tabular}{@{}lllllll@{}}
\toprule
$\mathcal{A}$ & $A_1$ & $A_{200}$ & $A_{400}$ & $A_{600}$ & $A_{800}$ & $A_{1000}$\\
\midrule
  $v_i$  & $1$ & $11,899$ & $39,809$ & $83,711$ & $143,610$ & $219,490$\\
    $c_i$  & $1$ & $3,981$ & $7,981$ & $11,981$ & $15,981$ & $19,981$\\
\botrule
\end{tabular}
\end{table}

By running \Cref{alg:3}, we get $m = 16$. There are two types of growth rate of the computational complexity, one is in exponential way that is upper bounded by $2^{n-2}$, and the other is in polynomial way that is less than $\frac{m^2}{n}\binom{n}{m}$. To show the efficient computability scheme, we select the first $2m$, $10m$, $\cdots$, $60m$ airplanes to compose an ARP instance respectively, and calculate the index number $\widetilde{m}$ by \Cref{alg:4}. Numerical results are displayed in \Cref{tab:7}.

\begin{table}[htbp]
\caption{Comparison between different $Q_n$.}
	\label{tab:7}       
\begin{tabular}{@{}lllll@{}}
\toprule
     $n$   &   $\widetilde{m}$  &  $\frac{\widetilde{m}^2}{n}\binom{n}{\widetilde{m}}$  &  $\frac{m^2}{n}\binom{n}{m}$  &  $2^{n-2}$   \\
\midrule
			  $32$   & $13$ &  $1.83 \times 10^{9}$  &  $4.81 \times 10^{9}$ & $1.07 \times 10^{9}$\\
              $160$   & $10$ & $1.42 \times 10^{15}$  &  $6.50 \times 10^{21}$ & $3.65 \times 10^{47}$\\
  	          $320$   & $10$ & $8.41\times 10^{17}$  &  $3.16 \times 10^{26}$  &  $5.34 \times 10^{95}$\\
              $480$   & $10$ &  $3.39 \times 10^{19}$  &  $1.57 \times 10^{29}$  &  $7.80 \times 10^{143}$\\
              $640$   & $10$ &  $4.63 \times 10^{20}$  &  $1.25 \times 10^{31}$  &  $1.14 \times 10^{192}$\\
              $800$   & $10$ & $3.50 \times 10^{21}$  &  $3.70 \times 10^{32}$  &  $1.67 \times 10^{240}$\\
  			  $960$   & $10$ & $1.82 \times 10^{22}$  &  $5.85 \times 10^{33}$  &  $2.44 \times 10^{288}$\\
             $1,000$  & $10$ &  $2.63 \times 10^{22}$  &  $1.08 \times 10^{34}$  &  $2.68 \times 10^{300}$\\
\botrule
		\end{tabular}
\end{table}

A further explanation of \Cref{thm:6} is that given a relatively large set of airplanes, if we have found the index $\bar{m}$, then for any $n$-airplane refueling instance drawn from the set of airplanes, the inflection point must be less than $2\bar{m}$, and the computational complexity is less than $\frac{\bar{m}^2}{n}\binom{n}{\bar{m}}$. For example, we choose $500$ airplanes from the $1000$ airplanes, and we get $m = 10$ (which is less than $\bar{m} = 16$). The computational complexity comparison results are presented in \Cref{tab:8}.

\begin{table}[htbp]
\footnotesize
\caption{Complexity comparison results}
	\label{tab:8}       
\begin{tabular}{@{}ll@{}}
\toprule
			$n = 500$ &  upper bound of $Q_n$  \\
\midrule
			  $Q_n \leqslant 2^{n-2}$  & $8.18 \times 10^{149}$    \\
                                       &                           \\
		      $Q_n < \frac{m^2}{n}\binom{n}{m}$ for $m = 16$  & $2.93 \times 10^{29}$    \\
                                     &                           \\
	          $Q_n < \frac{m^2}{n}\binom{n}{m}$ for $m = 10$  & $4.92 \times 10^{19}$    \\
                                     &                           \\
   	          $\widetilde{Q_n} < \frac{\widetilde{m}^2}{n}\binom{n}{\widetilde{m}}$ for $\widetilde{m} = 6$  & $1.17 \times 10^{16}$    \\
\botrule
		\end{tabular}
\end{table}

\section{Conclusion}

We proposed SSA to solve the ARP instances, and posed an efficient computability scheme. We found that the computational complexity of SSA running on ARP will decline from an exponential level to a polynomial level when its input size of $n$ is greater than an inflection point $N$. For each ARP instance whose airplane's single flight distance is located in an interval $[L, S]$, there must exist an inflection point $N$ such that SSA is a polynomial-time algorithm when the input size of $n$ is greater than an inflection point $N$. The upper bound of inflection point is determined by the interval $[L, S]$, and it dose not depend on the input size of $n$. Moreover, we proved that SSA is a polynomial-time algorithm to solve ARP. At last we constructed an efficient computability scheme to forecast in polynomial time the particular running time of SSA running on any given ARP instance, according to which we could provide computational strategy for decision makers and algorithm users.

%
%
%
%
%
%
%
%
%
%
%

\bibliography{references}

\end{document}